\documentclass[preprint,authoryear,12pt]{article}
\usepackage{mathrsfs}
\usepackage{natbib,graphicx,setspace,lscape,longtable}
\usepackage{stmaryrd}
\usepackage{natbib,epsfig,graphicx}
\usepackage{mathrsfs,amsmath,amsthm,amssymb,color}
\usepackage{enumerate}
\usepackage{multirow}
\usepackage{caption,subcaption}
\usepackage{varwidth}
\usepackage{algorithm}
\usepackage{algpseudocode}
\usepackage{tabularx,threeparttable,multirow,multienum,booktabs}
\usepackage{url}
\usepackage{bbm}

\newtheorem{theorem}{Theorem}

\newtheorem{lemma}{Lemma}
\newtheorem{proposition}{Proposition}

\newtheorem{corollary}{Corollary}
\newcommand{\bm}[1]{\mbox{\boldmath{$#1$}}}

\def\beq{\begin{equation}}
\def\eeq{\end{equation}}
\def\beqr{\begin{eqnarray}}
\def\eeqr{\end{eqnarray}}
\def\beqrs{\begin{eqnarray*}}
\def\eeqrs{\end{eqnarray*}}
\def\bet{\begin{theorem}}
\def\eet{\end{theorem}}
\def\bec{\begin{corollary}}
\def\eec{\end{corollary}}
\def\bel{\begin{lemma}}
\def\eel{\end{lemma}}
\def\bep{\begin{proposition}}
\def\eep{\end{proposition}}
\def\bg{\begin{figure}[tbph]\begin{center}}
\def\eg{\end{center}\end{figure}}

\def\bc{\begin{center}}
\def\ec{\end{center}}

\def\mR{\mathbb{R}}

\def\var{\mathrm{var}}
\def\cov{\mathrm{cov}}

\numberwithin{equation}{section}

\begin{document}

\title{Hypothesis Testing of One Sample Mean Vector  in Distributed Frameworks}
\author{Bin Du$^1$, \quad Junlong Zhao$^1$\footnote{Corresponding author. Email address: zhaojunlong928@126.com.
}
 \\
\small\it 1. School of Statistics, Beijing Normal  University, China\\
}
\maketitle
\begin{abstract}
\noindent    Distributed frameworks are widely used to handle massive data, where sample size $n$ is very large, and  data are often stored in $k$ different machines. For a random vector $X\in\mR^p$ with expectation $\bm\mu$, testing the mean vector $H_0: \bm\mu=\bm\mu_0$ vs $H_1: \bm\mu\ne \bm\mu_0$ for a given vector $\bm\mu_0$ is a  basic  problem  in statistics. The centralized test statistics require heavy communication costs, which can be a burden when $p$ or $k$ is large.
To reduce the communication cost,  distributed test statistics are  proposed in this paper for this problem based on the  divide and conquer technique, a commonly  used approach for distributed statistical inference.  Specifically, we extend two commonly used centralized test statistics to the distributed ones, that apply to    low and high dimensional cases, respectively.      Comparing the power of centralized test statistics and the distributed ones, it is observed that there is a fundamental tradeoff between communication costs and  powers of the tests. This is quite  different from the application   of the divide and conquer technique in many other problems such as estimation, where the associated distributed statistics can be as good as the centralized ones.
   Numerical results confirm the theoretical findings.

\noindent{\bf Key words:}\  mean test, large and high dimension, divide and conquer, power function.

\end{abstract}

\newpage
\setcounter{page}{1}
\section{Introduction}

In recent years, massive data are commonly encountered in many applications and  are often stored in $k$ different machines in a distributed system. Developing communication-efficient  statistical methods has attracted a lot of attention recently. The divide and conquer technique  is a popular method in distributed frameworks, where one constructs a statistic or an estimator using data in each machine, and then transmits them to the hub to  get a pooled one.
The divide and conquer technique  has been applied successfully  in many problems, including  regression and classification, hypothesis testing,  confidence intervals, principal eigenspaces analysis, linear discriminant analysis, and many others \citep[etc.]{Zhang2013,Hsieh2014,Zhang2015,Lin2017,Szabo2017,Battey2018,Guo2019,Chen2018,Michael2019, Fan2019, Tian2017,Limengyu2020,Dobriban2021}.

In many problems  such as point estimations, the statistic constructed by  the divide and conquer technique can be  as efficient  as the centralized one,  when the number of machines $k$ is not too large, usually much smaller than $\sqrt{n}$ \citep[etc.]{Chen2014,Zhang2015,Zhang2013,Michael2019,Volgushev2019,Chen2018}.
Although the divide and conquer technique has been used successfully in many problems in the distributed system, its application to  hypothesis testing of the mean vector for massive data has not  been studied.

For a variable   $X\in\mR^p$ with $E(X)=\bm\mu$, the hypothesis testing on the  mean vector $\bm\mu$    is a basic problem  in statistics, playing   a critical role  in many applications  such as quality control,   environmental science, geography, medicine, education, social sciences, and many others. For example, in quality control, to  test  a batch of  products is qualified or not, one needs to consider the test problem:  $H_0: \bm\mu=\bm\mu_0$ vs $H_1:\bm\mu\ne \bm\mu_0$, where $\bm\mu\in\mR^p$ is the mean of the products considered and  $\bm\mu_0\in\mR^p$ is the standard
 \citep[etc.]{Edward1985,Ye2002}.
 In addition, to analyze the difference of paired samples  and  the effect of treatments with matched samples \citep[etc.]{Davison1992,Rubin2006,Haug2011}, it is common to   transform the problem  into a one-sample mean test problem $H_0:\bm\mu=\bm0$ vs $H_1: \bm\mu\ne \bm0$, where  $\bm\mu\in\mR^p$ is the mean of the difference of the paired samples, and  rejecting $H_0$ implies that there exist  treatment effects.

In this paper, we consider the  hypothesis   testing problem
\begin{equation}\label{eq1}
 H_0: \bm\mu= \bm\mu_0 \quad versus \quad H_1: \bm\mu\neq\bm\mu_0,
\end{equation}
where $\bm\mu_0=(\mu_{01},\cdots,\mu_{0p})^\top$ is a given vector.
   Supposing that   $\{X_i\}_{i=1}^{n}$ are  $p$-dimensional random vectors that are  $i.i.d.$ from the normal distribution $N_p(\bm\mu,\Sigma)$ with the unknown covariance matrix  $\Sigma$, satisfying  $p<n$, the classical  test statistic for this problem  is the Hotelling $T^2$  statistic \citep{Anderson1984}, which is a centralized one,  defined as follows
\[T_{cen,n}^2=(n-1) (\bar X - \bm\mu_0 )^\top \hat\Sigma^{-1} (\bar X - \bm\mu_0 ),\]
where $\bar X=n^{-1}\sum_{i=1}^n X_i$  is the sample mean and $\hat\Sigma=\sum_{i=1}^n (X_i-\bar X)(X_i-\bar X )^\top/n$ is the sample covariance matrix.
However, when the sample covariance matrix is nearly singular, the power of the Hotelling $T^2$ test decreases significantly  even when $p < n$,
 \citep{Bai1996}. Many  extensions of the Hotelling
$T^2$ test to high dimensional cases  have been proposed in
the literature. For example, \cite{Bai1996} considered the case of  $p/n\to r\in(0,1)$, \cite{Srivastava2008} investigated the case  of $p=O(n^\zeta)$ with $\zeta\in (1/2,1)$, and \cite{Lee2012} studied the case of $p/n\to r>0$. In addition, many authors  considered
the case where  $p$ can be much larger than $n$ and the distribution is not Gaussian \citep[etc.]{Chen2010,Wang2015,Srivastava2016,Xu2016,Dong2016,Chakraborty2017,Li2020}. For example, the test statistic in \cite{Wang2015} has significant gains for heavy-tailed multivariate distributions.

Although there are many works on the problem (\ref{eq1}),  to the best of our knowledge, this problem  has not been studied in a distributed framework, when  massive data are  collected and are stored in $k$ different machines.  When $p<n$, the classical Hotelling $T^2$ test can be computed   by transmitting   $k$ matrices of the size $p\times p$, as shown in Section 2.1, of which  the   communication cost can be expensive  when $p$ or $k$ is large.
Taking this  into account, we propose a distributed version of the Hotelling $T^2$ test statistic that is  communication-efficient, based on the idea of  the divide and conquer technique, and  derive  the asymptotical distribution of the distributed test statistic under $H_0$.  Furthermore, we compare the power functions of the centralized Hotelling $T^2$ test with the distributed one.

When $p$ is much larger than $n$,  as an illustration,  we extend the centralized test statistic of \cite{Wang2015}, proposing a  distributed test statistic, deriving  its asymptotic distribution   under $H_0$, and compring the power functions of the distributed test with the centralized one.  Theorectical results show that the powers of the distributed test statistics are decreasing  functions of $k$, indicating that there is a fundamental tradeoff between communication cost and efficiency. This  is different from the existing results in other settings  such as point estimations, where the distributed statistics can be as good as the centralized ones.

The rest of the paper is organized as follows. In Section 2, a  distributed  test statistic  is constructed by applying the divide and conquer technique to the Hotelling $T^2$ test, and the  asymptotic distribution under $H_0$   is developed and the  power function of the distributed test statistic  is   compared with the centralized one.    In Section 3, we consider the high dimensional case, constructing a distributed test  statistic by applying the divide and conquer technique to the test statistic  of \cite{Wang2015}; the asymptotic distribution under $H_0$ and its power are studied. Simulation results and real data analysis are reported  in Section 4 and a brief discussion is presented in Section 5. All the proofs are  presented in Section  6.

\section{One Sample  Mean Test  in  Distributed Frameworks}
\subsection{Computation of the centralized Hotelling $T^2$ test in  distributed frameworks}

Since  data are stored in $k$ different machines, we denote by $S_l$ the set of  indices of the observations on the $l$-th machine, that is, $S_l=\{i:X_i$  is stored in the $l$-th  machine$\}$, $l=1,\cdots,k$. Let  $n_l=|S_l|$, the sample size on the $l$-th machine.
In the distributed framework,   constructing the centralized test statistic  is still feasible by transmitting some matrices of size $p$ by $p$, as shown in Algorithm 1.

\begin{algorithm}[ht]
\caption{Computing the Centralized Hotelling $T^2$ in a Distributed  Framework}
 \begin{algorithmic}[1]
  \State For $l=1,\cdots, k$,  compute the local sample mean $\bar X_l$ and the local second moment   $A_l$ using  observations in the $l$-th machine where
      $\bar X_l=n_l^{-1} \sum_{i\in S_l}X_i$ and $A_l=n_l^{-1}\sum_{i \in S_l} X_iX_i^\top.$

  \State Transmit the $\bar X_l$'s and $A_l$'s to the central hub, and compute the  global sample
  mean and the global sample covariance matrix denoted as $\bar X=n^{-1}\sum_{\mathnormal{l}=1}^k n_l \bar X_l$ and  $ \hat\Sigma=n^{-1}\sum_{l =1}^k  n_lA_l -\bar X\bar X^\top$, respectively.

  \State Define the  centralized Hotelling $T^2$  test statistic $T_{cen,n}^2=(n-1)(\bar X -  \bm\mu_0 )^\top \hat\Sigma^{-1}(\bar X- \bm\mu_0 )$.

 \end{algorithmic}
\end{algorithm}

Suppose  that $\{X_i\}_{i=1}^{n}$ are $i.i.d.$ random variables  following $N_p(\bm\mu,\Sigma)$.  Under $H_0$,  according to the relationship between Hotelling $T^2$ distribution and the $F$-distribution, we have
\[ \frac{n-p}{(n-1)p} T_{cen,n}^2 \sim F(p,n-p), \]
where $F(p,n-p)$ denotes the $F$-distribution with the degrees of freedom $p$ and $n-p$.
Given the significant level $\alpha\in(0,1)$,  the rejection region for Hotelling $T^2$ test is denoted as
\[\mathcal{C}_1=\Big\{T_{cen,n}^2:~\frac{n-p}{(n-1)p} T_{cen,n}^2 > F_{1-\alpha} (p,n-p)\Big\},\]
where $F_{1-\alpha} (p,n-p)$ is $1-\alpha$ quantile of $F(p,n-p)$.
 Clearly, to obtain centralized test statistic $T_{cen,n}^2$, one  needs to transfer $p\times p$ matrices $A_l$'s to the hub, which can be a burden   when $p$ or $k$ is large.  To reduce the communication cost, we consider a distributed   test statistic.
\subsection{Extending   Hotelling $T^2$ test in distributed frameworks}

 For simplicity of notations, assume that data are randomly and evenly distributed  in $k$ machines, that is
 $$n_1=\cdots=n_k=n/k,$$
  where  $n_l$ is  the sample size in the $l$-th machine, $l= 1, \cdots, k$.
  We propose the following  distributed test statistic using the divide and conquer technique, by
       computing first the  Hotelling $T^2$ test statistic with data  in each machine, and then pooling  them together.  Specifically,   the distributed test statistic is defined as follows
$$T_{dis,n}^2=\frac{1}{k\sqrt{p}}  \sum_{l =1}^k T_l ^2,$$
   where $T_l ^2 =(n_l -1)(\bar X_l - \bm\mu_0)^\top \hat\Sigma_l ^{-1}(\bar X_l - \bm\mu_0)$ is the local Hotelling $T^2$ statistic based on the data in the $l$-th machine, and $\hat\Sigma_l$ and $\bar X_l$ are  the local sample covariance matrix and the local sample mean, respectively.
 Note that $T_{dis,n}^2$ is communication efficient, requiring  only the transmission of the scalars $T_l$'s to the hub.
Under $H_0$, it is shown in the Theorem \ref{Th1} that   $T_{dis,n}^2$  converges  to a normal distribution.  Assume that  $k$ and $p$ are functions of $n$, that is, $k =k_n$ and $ p=p_n$, and the following conditions hold.
\begin{itemize}
  \item[(A1)] Let  $\gamma_{n}=pk/n$. Assume that $\gamma_n<1$ and that $\gamma_{n}\to r\in[0, 1)$, as $n\to \infty$. Moreover, assume that   $n_l=n/k\to\infty, l=1,\cdots,k$, as $n\to \infty$.
\end{itemize}
Condition (A1)  requires that $k=O(n/p)$, where both $k$ and $p$ can diverge with $n$.
\bet\label{Th1}
Suppose that $X_i$'s are i.i.d. from $N_p(\bm\mu,\Sigma)$ and  that condition (A1) holds.  Under $H_0$,  as $k\to\infty$,  it holds that
\begin{equation}\label{T_p-norm}
\sqrt k\left(T_{dis,n}^2-\frac{(n/k-1)\sqrt p}{n/k-p-2}\right)\stackrel{d}{\longrightarrow} N\left(0,\frac 2 {(1-r)^3}\right).
\end{equation}
\eet

Based on Theorem \ref{Th1}, the rejection region for the  distributed test can be written as follows,
 \[\mathcal{C}_2=\Big\{ T_{dis,n}^2:~\left(\frac{(1-r)^3 k}{2}\right)^{1/2} \Big|  T_{dis,n}^2-\frac{(n/k-1)\sqrt p}{n/k -p-2}\Big| > z_{1-\frac \alpha 2}\Big\},\]
where $z_{1-\alpha/2}$ is the $1-\alpha/2$ quantile of standard normal distribution.
Compare the distributed test statistic with the centralized one, there are two observations.
(1) To compute the  centralized test statistic, one  needs to transmit  $k$ matrices of size $p\times p$ to the hub, while the distributed one only requires the  transmission  of  $k$ scalars to the hub, having great advantages in communication cost.
(2)   For the distributed test statistic, we get a normal distribution asymptotically, which differs from the centralized one.

\subsection{The power of the distributed test}

In this section, we compare  the powers of the distributed test statistic and the  centralized one  under $H_1$.
Denote by  $\psi_n^{cen}=P( \mathcal{C}_1| H_1 )$
and     $\psi_n^{dis}=P( \mathcal{C}_2| H_1)$ the  power functions of the centralized test and the distributed one, respectively.    Let $\alpha$ be a significance level of the test and let  $\Delta=( \bm\mu- \bm\mu_0)^\top\Sigma^{-1}( \bm\mu- \bm\mu_0)$.
Denote
 \begin{align}
\phi_{\Delta,\gamma_{n}}=1&-\Phi\left(A_{1n} z_{1-\frac{\alpha}{2}}- A_{2n}\sqrt{(1-\gamma_{n})/2} \right)+\Phi\left(A_{1n} z_{\frac{\alpha}{2}}- A_{2n}\sqrt{(1-\gamma_n)/2} \right)\notag
\end{align}
with
$$A_{1n}=\left(\frac{\gamma_{n}}{\gamma_{n}+2\Delta+\Delta^2}\right)^{1/2}, \qquad A_{2n}=\left(\frac {n\Delta^2}{\gamma_{n}+2\Delta+\Delta^2}\right)^{1/2},  $$
 where  $\Phi(\cdot)$ denotes   the cumulative distribution function of standard normal distribution. The properties of $\psi_n^{cen}$ and $\psi_n^{dis}$ are as follows. In the following descriptions, notations $a_n\gg b_n$ and $a_n\ll b_n$ mean $a_n/b_n\to \infty$ and $a_n/b_n\to 0$ respectively, as $n\to \infty$.
\bet\label{Th2}
(1) If $\Delta\gg\sqrt p/n$, then $\psi_n^{cen}\to1$ when $n\to\infty$.  (2) Assuming that condition (A1) holds and that $k\to\infty$, then
   $\psi_n^{dis}-\phi_{\Delta,\gamma_{n}}\to 0$.
\eet
The proof of Theorem \ref{Th2} is given  in Appendix.   From the definition of $\phi_{\Delta,\gamma_n}$, it is easy to see  that $\phi_{\Delta,\gamma_{n}}\to 1$, when $A_{1n}\to 0$ or $ A_{2n}\to \infty$. The following corollary clarifies conditions of $\psi_n^{dis}\to1$.
\begin{corollary}\label{Cor}
Assume  that condition  (A1) holds and that $k\to\infty$.  Then
  $A_{1n}\to 0$ when  $\Delta\gg \gamma_n$, and $A_{2n}\to\infty$ when  $\Delta\gg\sqrt{kp}/n$.
Consequently, $\psi_n^{dis}\to1$ as $\Delta\gg\sqrt{kp}/n$.
\end{corollary}
\begin{proof}
The conclusions in Corollary \ref{Cor} is derived by discussing the order of $\Delta$ and $\gamma_n$.
First, we prove the conclusion on $A_{1n}$.
From the definition of $A_{1n}$, it holds that  $A_{1n}\to0$ when $\Delta\gg\gamma_n$ or $\Delta^2\gg\gamma_n$. Because $\gamma_n<1$, we see that  $A_{1n}\to0$ when $\Delta\gg\gamma_n$.
Second, we prove the conclusion on  $A_{2n}$. Recall that
\[A_{2n}=\left(\frac {n\Delta}{\gamma_{n}/\Delta+\Delta+2}\right)^{1/2}.\]
When $\Delta\gg\gamma_n$, it is seen    that   $A_{2n}\gg \sqrt{kp}$, due to $kp<n$. Moreover, when $\Delta$ and $\gamma_n$ have the same order, $A_{2n}$ has the same order as $\sqrt{kp}$. Due to $k\to\infty$,  we see that  $A_{2n}\to\infty$, when $\gamma_n=O (\Delta)$.
When $\Delta\ll\gamma_n$, it is easy to see that
\[A_{2n}=\left(\frac {n\Delta^2}{\gamma_{n}}\right)^{1/2}(1+o(1)),\]
and consequently that  $A_{2n}\to\infty$, when $\Delta\gg\sqrt{\gamma_n/n}$, i.e. $\Delta\gg\sqrt{kp}/n$.
Combining the argument together, we have $A_{2n}\to\infty$, when  $\Delta\gg\sqrt{kp}/n$.
\end{proof}

To get some insights, we make a plot  of the function  $\phi_{\Delta,\gamma_{n}}$ with $n=10^5$, $p=10^3$ and $\alpha=0.05$ for different values of $k$  in Figure \ref{fig1}. It is seen that   $\phi_{\Delta,\gamma_{n}}$    is an  increasing function of $\Delta$ and a decreasing one of $k$. And  from Theorem \ref{Th2} and Corollary \ref{Cor}, we notice the power of the distributed test is lower than the centralized one.
\begin{figure}[htb]
\centering
\includegraphics[width=8cm,height=6cm]{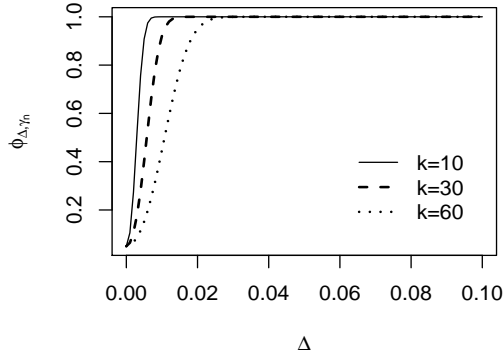}
\caption{Plot of the function $\phi_{\Delta,\gamma_n} $ for different $k$ and $\Delta$.  \label{fig1}}
\end{figure}

Recall that $\psi_n^{dis}\to1$ as $\Delta\gg\sqrt{kp}/n$ and that  $\psi_n^{cen}\to1$ as $\Delta\gg\sqrt p/n$. 
Then the asymptotic relative efficiency of the distributed test statistic to the centralized one can be measured by    the ratio   $(\sqrt p/n)/(\sqrt{kp}/n)=\sqrt{1/k}$. 

\section{ One Sample Mean Test in  High-dimensional Cases in Distributed Frameworks}
Hotelling $T^2$ test statistic is not well defined when the dimension is larger than the sample size. Moreover, the distribution of Hotelling $T^2$ test statistic depends on the normal distribution heavily. In this section, we apply  the divide and conquer approach to the  test statistics that are applicable to general distributions and high dimensional settings.

\subsection{Review of some extensions of Hotelling $T^2$ statistic in high dimensional cases}
Although Hotelling $T^2$ test is a classical one, it  has serious  limitations when the  dimension $p$ is comparable to, or larger than the sample size $n$. Obviously, the test statistic is not well defined for $p>n$ because of the singularity of the sample covariance matrix. Moreover, Hotelling $T^2$ test    is also known to perform poorly when  $p/n\to r\in (0,1)$.
 For example, \cite{Dempster1958} proposed the non-exact test in high dimension.
 \cite{Bai1996} further found that even when Hotelling $T^2$ statistic is well defined, Dempster's non-exact test is more powerful than the Hotelling $T^2$ for $p/n\to r\in (0,1)$.

To overcome the limitations of Hotelling $T^2$ test.  Many test statistics are developed by  replacing the inverse of sample covariance matrix in Hotelling $T^2$ statistic with a diagonal matrix or the identity matrix. These statistics are  applicable to high-dimensional cases and to more general distributions beyond normal.
For example, \cite{Bai1996}  assumed the following model: $X_{i}=\Lambda W_{i}+\bm\mu_{i}$, for $i=1,\cdots,n$, where $\Lambda$ is a $p\times m$ non-random matrix for some $m\ge p$ such that $\Lambda\Lambda^\top=\Sigma$, and $\{W_{i}\}_{i=1}^{n}$ are $m$-dimensional vectors   satisfying $E(W_{i})=0$ and ${\rm cov}(W_{i})=I_m$, the $m\times m$ identity matrix. And they assumed $p/n\to r>0$, that is, $p$ can be comparable to $n$ or larger than $n$. This model  includes the normal distribution as a special case.
Instead of using Mahalanobis distance, \cite{Bai1996} constructed the test statistic based on Euclidean distance as follows,
$$
T_{BS}=n(\bar X-\bm\mu_0)^\top(\bar X-\bm\mu_0)-\frac{n}{n-1}{\rm Tr}(\hat\Sigma),
$$
where $\bar X$ is the sample mean. %
\cite{Srivastava2008} noticed that the statistic $T_{BS}$ is not scale-invariant, which may reduce the power when the scales of different components of the random vector $X_i$ are quite different. As a remedy, they proposed  the following statistic
$
T_{SD}=(\bar X-\bm\mu_0)^\top D^{-1}(\bar X-\bm\mu_0),
$
where $D={\rm diag}(n\hat\sigma_{11}/(n-1),\cdots,n\hat\sigma_{pp}/(n-1))$, consisting of the diagonal elements  of the matrix $n\hat\Sigma/(n-1)$ with $\hat\Sigma$ denoting the sample covariance matrix. In addition,  it is assumed by \cite{Srivastava2008}  that $n=O(p^\zeta)$ for some $1/2<\zeta\le1$ with normality assumption. Furthermore,  \cite{Srivastava2009} relaxed the  normality assumption, and assumed $n=O(p^\zeta)$, $0<\zeta\le1$.
\cite{Chen2011} proposed another test statistic
$$T_{RHT}=(\bar X-\bm\mu_0)^\top (n\hat\Sigma/(n-1)+\lambda I)^{-1}(\bar X-\bm\mu_0),\quad \text{for $\lambda>0,$}$$
under the setting $p/n\to r\in(0,\infty)$,   stabilizing  the inverse of the sample  covariance matrix.
 As an improvement to  the work of  \cite{Bai1996}, \cite{Chen2010} proposed a statistic by removing the  quadratic   terms in $(\bar X-\bm\mu_0)^\top(\bar X-\bm\mu_0)$, which relaxed the restrictions on $p$ and $n$ greatly. \cite{Wang2015} considered robust testing procedures which are powerful to distributions with heavy-tailed and to data with outlying observations.
  In the following subsection 3.2, as an example for illustration, we apply the  divide and conquer method  to the test statistic of \cite{Wang2015}.

\subsection{Nonparametric multivariate tests in a distributed system}
To handle heavy-tailed distributions,   \cite{Wang2015}    assumed  that  $X_i$'s are $i.i.d.$ variables  following  an elliptical distribution. Specifically,
\begin{equation}\label{eq5}
X_i=\bm\mu+\epsilon_i, \quad \epsilon_i=\Gamma R_iU_i,\qquad 1\le i\le n,
\end{equation}
where $\Gamma$ is a $p\times p$ matrix, $U_i$ is a random vector uniformly distributed on the unit sphere in $\mR^p$, and $R_i$ is a nonnegative random variable independent of $U_i$. 
For the hypothesis testing (\ref{eq1}),  the centralized test statistic in \cite{Wang2015}
is as follows
$$G_{cen,n}=\sum_{i=1}^n\sum_{j=1}^{i-1}Z_i^\top Z_j, $$
where $Z_i=(X_i-\bm\mu_0)/\|X_i-\bm\mu_0\|$ if $X_i\neq\bm\mu_0$, $Z_i=0$ if $X_i=\bm\mu_0$, and $\|X_i\|$ denotes the $L_2$ norm of $X_i$.
When data are stored at random in $k$ machines, $G_{cen,n}$ can be computed  as follows. For $l=1,\cdots,k$, we compute local sample sum $B_{l}=\sum_{i\in S_l} Z_i$ and transmit $B_l$'s to the central hub. Then $G_{cen,n}=[(\sum_{l=1}^k B_l)^\top(\sum_{l=1}^k B_l)-n]/2$.

However, when $p$ is large and communication cost is prohibitively heavy, we are interested in the following distributed version of the statistic
$$G_{dis,n}=\sum_{l=1}^k\sum_{\substack{ i,j\in S_l,\\ j<i}}Z_i^\top Z_j:= \sum_{l=1}^k G_{l},$$
where  $G_{l}=\sum\limits_{i,j\in S_l, j<i}Z_i^\top Z_j\in \mR$ and $S_l$ is the indices of observations in $l$-th machine. To obtain $G_{dis,n}$,  only  scalars $G_{l}$'s need to be transmitted to  the hub.  Comparison of the  powers of $G_{cen,n}$ and $G_{dis,n}$ is an interesting problem.
To derive the asymptotic properties of  $G_{cen,n}$ and $G_{dis,n}$, we need the following conditions, where ${\rm Tr}(\cdot)$ denotes the trace and $\Sigma=\cov(X_i)$.
\begin{itemize}
  \item [(A2)] Let $a_{m,\Sigma}=\mathrm{Tr}(\Sigma^m)$ for any positive integer $m$, and $a_{\max}$ and $a_{\min}$ be the maximum and minimum eigenvalues of $\Sigma$, respectively. Assume that   ($i$) $a_{4,\Sigma}=o(a_{2,\Sigma}^2)$; ($ii$) $(a^4_{1,\Sigma}/a^2_{2,\Sigma})\exp \big\{-a^2_{1,\Sigma}/(128pa_{\max}^2) \big\}=o(1)$; ($iii$) $a_{\max}=o(a_{1,\Sigma})$;   ($iv$)
 $\exp \big\{-a_{1,\Sigma}^2/(256 p a^2_{\max} )\big\}=o\big(\min \{a_{\max}/ a_{1,\Sigma}, a_{\min }/a_{\max}\}).$
  \item[(A3)]  For some $0<\delta<1$, $\|\bm\mu-\bm\mu_0\|^{2 \delta} E\big(\|\epsilon_i\|^{-2-2 \delta}\big)= o\big(E^{2}\big(\|\epsilon_i\|^{-1}\big)\big) .$
\item[(A4)] $\|\bm\mu-\bm\mu_0\|^{2} E\big(\|\epsilon_i\|^{-2}\big)=o\big(\min \{a_{2,\Sigma}/(n a_{\max}a_{1,\Sigma}) ,  a^{1/2}_{2,\Sigma}/(n^{1 / 2}a_{1,\Sigma})\}\big) .$
\end{itemize}
These conditions are exactly the conditions (C1)--(C6) in \cite{Wang2015}.
Let $$ A_\epsilon=E((I_p-\epsilon_i\epsilon_i^\top/\|\epsilon_i\|^2)/\|\epsilon_i\|), \qquad  B_\epsilon=E(\epsilon_i\epsilon_i^\top/\|\epsilon_i\|^2).$$
 \cite{Wang2015} proved the following conclusion on the centralized test statistic $G_{cen,n}$.
\begin{lemma}(\cite{Wang2015})
   Assuming  the conditions ($i$) and ($ii$) in (A2), under $H_0$, as $n,p\to\infty$, it holds that
$$G_{cen,n}/\sqrt{n(n-1){\rm Tr}(B_\epsilon^2)/2}\stackrel{d}{\rightarrow}  N(0,1).$$
 And under $H_1$, if the conditions (A2)--(A4) hold,  as $n,p\to\infty$, $$[G_{cen,n}-n(n-1)(\bm\mu-\bm\mu_0)^\top A_\epsilon^2 (\bm\mu-\bm\mu_0 )(1+o(1))/2]/\sqrt{n(n-1){\rm Tr}(B_\epsilon^2)/2}\stackrel{d}{\rightarrow}  N(0,1).$$
\end{lemma}
To discuss the asymptotic properties of the distributed statistic $G_{dis,n}$, we assume
\begin{itemize}
   \item[(A4$'$)] $\|\bm\mu-\bm\mu_0\|^{2} E\big(\|\epsilon
_i\|^{-2}\big)=o\big(\min \{k a_{2,\Sigma}/(n a_{\max}a_{1,\Sigma}) ,  (ka_{2,\Sigma})^{1/2}/(n^{1 / 2}a_{1,\Sigma})\}\big) .$
\end{itemize}
 It is seen that (A4$'$) is weaker than (A4). The asymptotic  distribution of $G_{dis,n}$ is as follows.  

\begin{theorem} \label{thmh} Let $k=k_n$ and $p=p_n$ be constants depending on $n$ and set  $n_l=n/k$.
 Assume  that $\min\{k, n_l, p\}\to \infty$ as $n\to\infty$.
\begin{itemize}
  \item[(1)]
   Under $H_0$, assuming that conditions ($i$) and ($ii$) in (A2) hold, as $n\to\infty$,
 it follows that
 $$\frac{G_{dis,n}}{\sqrt{n(n/k-1){\rm Tr}(B_\epsilon^2)/2}}\stackrel{d}{\longrightarrow} N(0,1).$$

  \item [(2)]  Under $H_1$,
assuming that conditions (A2), (A3), and (A4$'$) hold, as $n\to \infty$,
 we have
$$\frac{G_{dis,n}-n(n/k-1)(\bm\mu-\bm\mu_0)^\top A_\epsilon^2 (\bm\mu-\bm\mu_0) (1+o(1))/2}{\sqrt{n(n/k-1){\rm Tr}(B_\epsilon^2)/2}}\stackrel{d}{\longrightarrow} N(0,1).$$
\end{itemize}
\end{theorem}
Under the condition ($i$) in (A2), when $\min\{n_1,p\}\to \infty$ , following \cite{Chen2010} and \cite{Wang2015},
 ${\rm Tr}(B_\epsilon^2)$ can be estimated using only the data in the first  machine, specifically, 
\begin{align}
\widehat{{\rm Tr}(B_\epsilon^2)}=&-\frac{n_1}{(n_1-2)^2}+\frac{(n_1-1)}{n_1(n_1-2)^2}{\rm Tr}\bigg\{\bigg(\sum_{j\in S_1} Z_j Z_j^\top\bigg)^2\bigg\}    \notag \\
&+\frac{1-2n_1}{n_1(n_1-1)}\bar Z^{*\top}\bigg(\sum_{j\in S_1} Z_j Z_j^\top \bigg)\bar Z^{*}  +\frac{2}{n_1}\|\bar Z^{*}\|^2+\frac{(n_1-2)^2}{n_1(n_1-1)}\|\bar Z^{*}\|^4,   \notag
\end{align}
where $\bar Z^{*}=(n_1-2)^{-1}\sum_{i\in S_1}Z_i$,
We  compare the centralized test with the distributed one in terms of the asymptotic relative efficiency. Let  $\eta_p=(\bm\mu-\bm\mu_0)^\top A_\epsilon^2(\bm\mu-\bm\mu_0)(1+o(1))/\sqrt{2{\rm Tr}(B_\epsilon^2)}$.
Theorem \ref{thmh} implies that under $H_1$, the distributed test  with sample size $n$ has the power
$$\psi_n^{dis}=1-\Phi(z_{1-\alpha/2}-\sqrt{n(n/k-1)}\eta_p)+\Phi(z_{\alpha/2}-\sqrt{n(n/k-1)}\eta_p).$$
On the other hand, by \cite{Wang2015},  the centralized statistic $G_{cen,N}$ with sample size $N$ has the power
$$\psi_N^{cen}=1-\Phi(z_{1-\alpha/2}-\sqrt{N(N-1)}\eta_p)+\Phi(z_{\alpha/2}-\sqrt{N(N-1)}\eta_p).$$
 Then we describe the asymptotic relative efficiency of the distributed test to the centralized one by the ratio of the sample sizes $N/n$ such that $\psi_N^{cen}=\psi_n^{dis}$, that is,   $N/n=1/\sqrt{k}$.
Thus,  the distributed test is  less efficient than that of  the centralized one. 

Combining the results in Section 2 and 3 together we derive the following conclusions. For the distributed test statistic constructed using  the divide and conquer method, there is a tradeoff between the efficiency represented by the  power  and the communication cost.  On the other hand, when point estimations are concerned, many works have shown that in terms of means square error (MSE), the estimators constructed using divide and conquer can be the same order as that of the centralized estimators \citep[etc.]{Chen2014,Zhang2015,Michael2019,Volgushev2019,Chen2018}.

\section{Numerical Experiments}
\subsection{Simulation results}
We illustrate the performance of the distributed  tests and compare them with that of the centralized ones by simulations. For any integer $m$, denote by $\mathbf{0}_m$ the vector of zero in $\mR^m$, and $\mathbf{1}_m$  the vectors in  $\mR^m$ with elements 1.
Let $\bm\mu_0=\mathbf{0}_p$, and consider the following hypothesis testing problem
$$ H_0: \bm\mu=\bm\mu_0~~\text{vs}~~ H_1: \bm\mu\neq\bm\mu_0.$$
For both distributed and centralized tests,  we repeat  500 times to  compute the type I error and the means of the power.

{\bf Example 1}. We consider the case of  $pk<n$. Let $\{X_i\}_{i=1}^{n}$ be $i.i.d.$ $p$-dimensional random vectors  from  the normal distribution  $N_p( \bm\mu,\Sigma)$, where  $E(X_i)=\bm\mu\in\mR^p$ and $\cov(X_i)=\Sigma\in\mR^{p\times p}$.  With the notation  $\tilde\Sigma=(\tilde\sigma_{ij})$ with   $\tilde\sigma_{ij}=0.5^{|i-j|}$, we assume that $\bm\mu$ and $\Sigma$  are generated from the
following four cases:   (1) $ \bm\mu=(c\mathbf{1}_2^\top, \mathbf{0}_{p-2}^\top)^\top$ and $\Sigma=I_p$;    (2) $\bm\mu=c\mathbf{1}_p$ and  $\Sigma=I_p$; (3) $ \bm\mu=(c\mathbf{1}_2^\top, \mathbf{0}_{p-2}^\top)^\top$ and  $\Sigma=\tilde\Sigma$;  (4) $\bm\mu=c\mathbf{1}_p$ and $\Sigma=\tilde\Sigma$, where    $c$ is a constant, and different values of $c$ will be considered.

 We set $p=50$ and 100 respectively. For the distributed test, we set $k=30,50$, and 70, respectively, and split the data   randomly into $k$   machines with equal sizes.  For each setting, we estimate the sizes of the test   by   the frequency of rejection in
 500 replicas.  Simulation results on the type I error for $k=30$ are presented in Table \ref{Table1}, and those of  $k=50$ and $k=70$ are similar and are omitted.

\begin{table}[ht]
\begin{center}
\begin{tabular}{c|cccc}
\hline
                             & \multicolumn{2}{c}{$\Sigma=I_p$} & \multicolumn{2}{c}{$\Sigma=\tilde\Sigma$} \\
     $n$                        &  $p=50$    &  $p=100$  &  $p=50$    &  $p=100$  \\ \hline
 5000  &       0.060      &   0.058        &       0.056      &   0.056          \\
  8000  &    0.048         &        0.052    &          0.048   &         0.050   \\
  10000 &       0.052      &         0.052   &          0.048   &          0.052  \\ \hline
\end{tabular}
\caption {Type I error of the distributed  test under $H_0$ for  $k=30$ with significant level $\alpha=0.05$ for Example 1.\label{Table1} }
\end{center}
\end{table}

The simulation results  with $n=10000$ and $p=50$ for different $k$ are presented in Figure \ref{fig2}, and results on $p=100$ are similar and are omitted.
From Figure \ref{fig2}, it is seen that the power of  distributed  test increases to 1, as $c$ increases. In addition, the distributed test is less powerful than the centralized one, which coincides with the theoretical results.

\begin{figure}[htb]
\begin{minipage}[t]{0.5\linewidth}
\centering
\includegraphics[width=8cm,height=6cm]{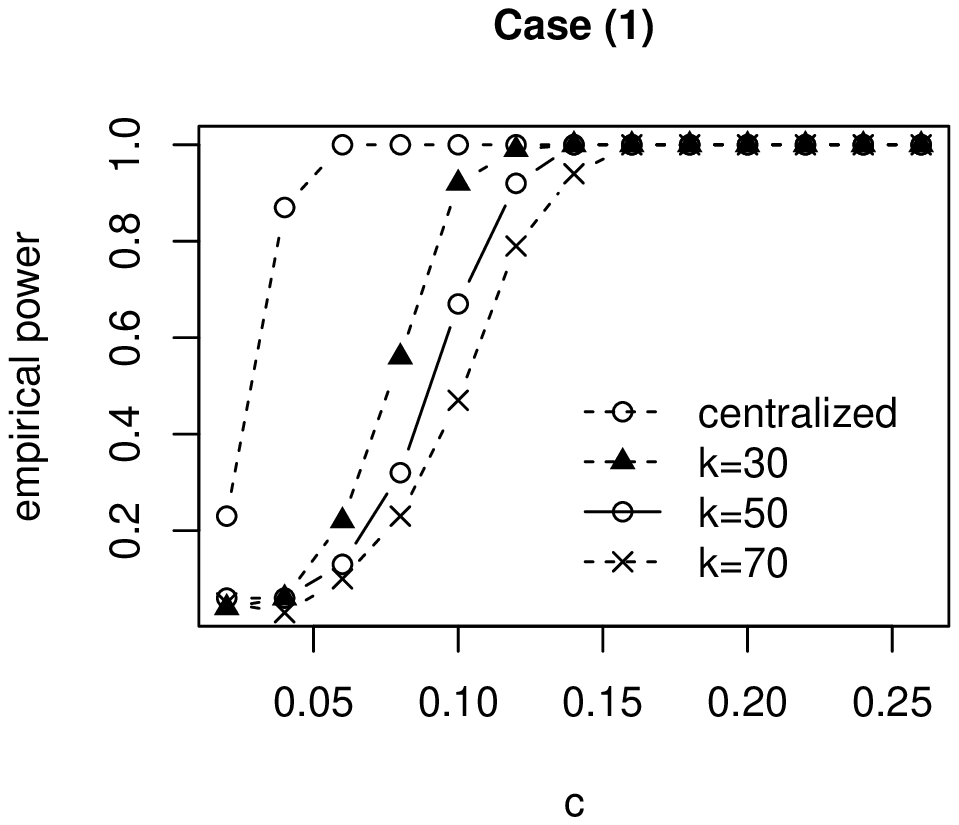}
\end{minipage}
\hfill
\begin{minipage}[t]{0.5\linewidth}
\centering
\includegraphics[width=8cm,height=6cm]{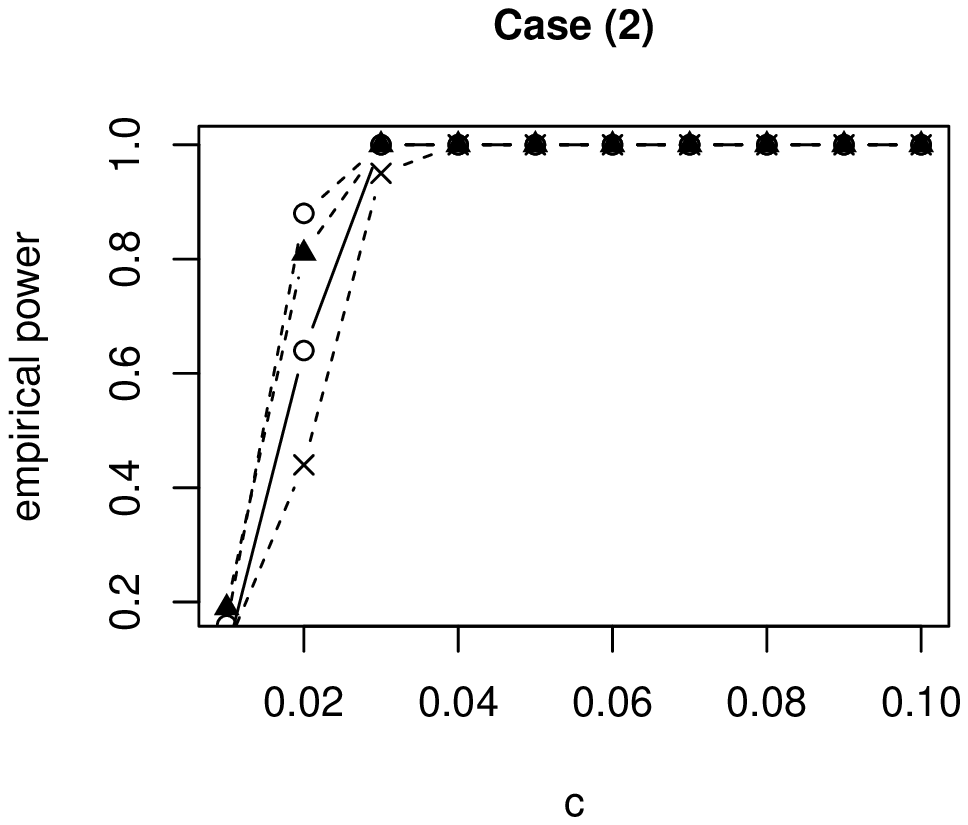}
\end{minipage}
\begin{minipage}[t]{0.5\linewidth}
\centering
\includegraphics[width=8cm,height=6cm]{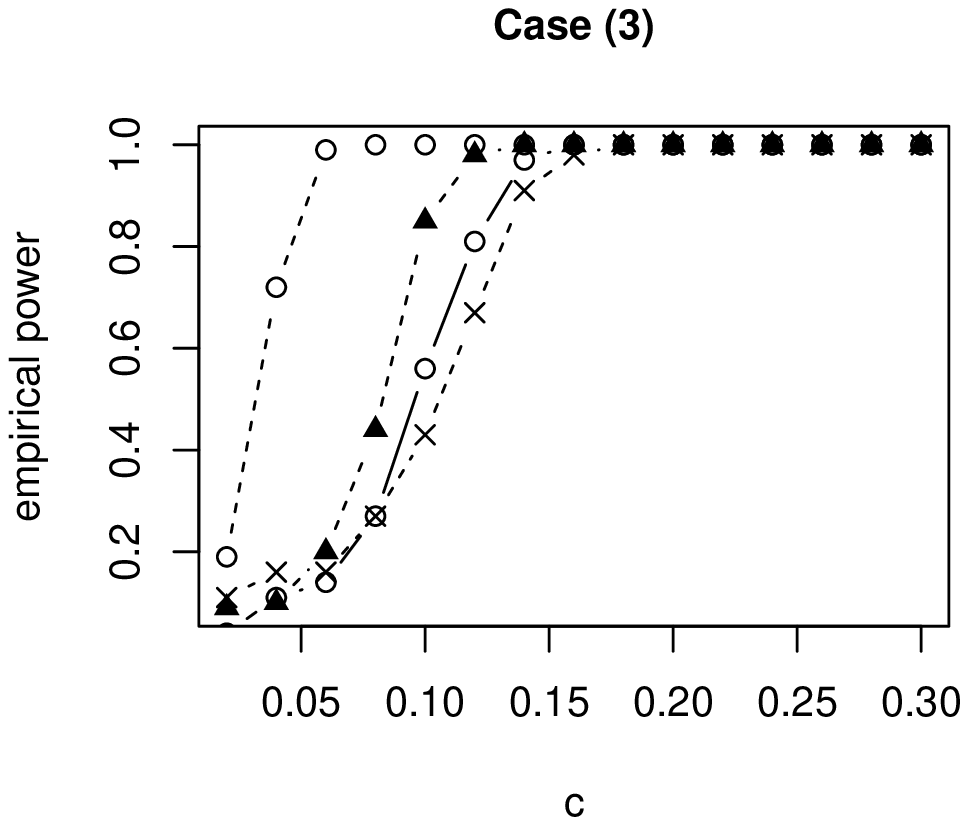}
\end{minipage}
\hfill
\begin{minipage}[t]{0.5\linewidth}
\centering
\includegraphics[width=8cm,height=6cm]{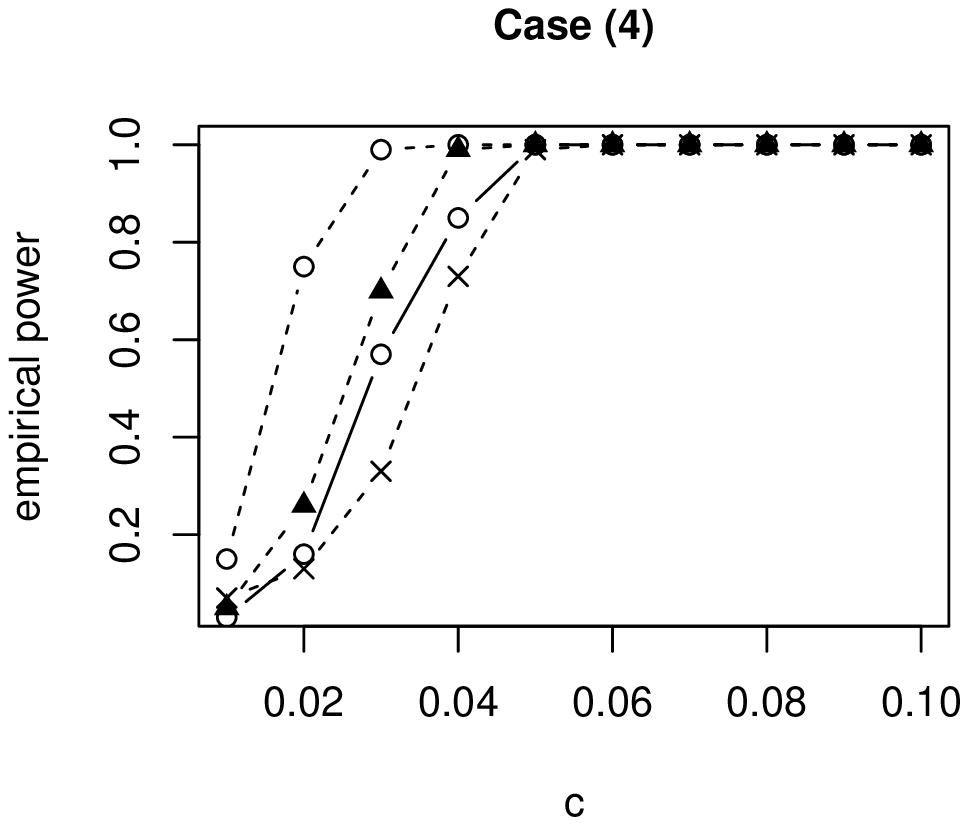}
\end{minipage}
\caption{Comparison of the powers of the centralized test with those of the distributed one in Example 1  for different $k$ with $n=10000$ and  $p=50$. \label{fig2}}
\end{figure}

{\bf Example 2}. We consider the high-dimensional cases where  $p>n$.  Let  $\{X_i\in\mR^p, 1\le i\le n\}$  be $i.i.d.$ sample  with the mean $E(X_i)=\bm\mu$ and covariance $\Sigma=(\sigma_{ij})$. Let   $p=1000$ and   $n=900$. Moreover, we take  $k=10,20$ and 30, and split the data   randomly into $k$   machines such that each machine has samples of size $n/k$, which  is  much smaller than $p$.    Let  $\bm\mu =(c\mathbf{1}_{20}^\top,\mathbf{0}_{p-20}^\top)^\top$ with $c=0,0.25,0.35,0.45$, respectively,  and  $c=0$ implies that  $H_0$ is true.
$\Sigma$ is taken as either $\Sigma_1$ or $\Sigma_2$, defined as    $\Sigma_1=(\sigma_{1,ij})$ with    $\sigma_{1,ij}=0.8^{|i-j|}$,  and  $ \Sigma_2=(\sigma_{2,ij})$ with $\sigma_{2,ii}=1$ and $\sigma_{2,ij}=0.2$ for $i\neq j$.
We consider two cases: (1) $X_i$ follows the normal distribution $N_p(\bm\mu,\Sigma)$.  (2) $X_i$ follows the  multivariate $t$ distribution $t_\nu(\bm\mu,\Sigma)$ with the degrees of freedom $\nu=3$.

Simulation results are presented in Table \ref{Table 2}. From Table \ref{Table 2}, it is seen that the type I error is controlled well for both methods.  In addition, we see that the powers of the distributed test statistics  decrease as $k$ grows and is lower than those of the centralized ones, which coincides with the theoretical findings.
\begin{table}
\begin{center}
\begin{tabular}{llllllllll}
\hline
          &   & \multicolumn{4}{c}{$X_i\sim N(\mu,\Sigma)$} & \multicolumn{4}{c}{$X_i\sim t_{3}(\mu,\Sigma)$} \\ \cmidrule(r){3-6}\cmidrule(r){7-10}
                   & $c$         & Cen   & $k=10$   & $k=20$   & $k=30$  & Cen   & $k=10$   & $k=20$   & $k=30$ \\ \hline
\multirow{4}{*}{}  $\Sigma=\Sigma_1$      &0   & 0.058      & 0.038     &  0.048    & 0.050     &0.050   & 0.046    & 0.038     &  0.030      \\
                   & 0.25  & 1      &0.926      &0.718      & 0.538  & 1 &0.928     & 0.740     &0.588     \\
                       &  0.35 &  1    &1& 0.996 &0.984   & 1      &1      &0.996      &0.980   \\
                       &  0.45 & 1      &  1    & 1     & 1  & 1      &  1    & 1     & 1   \\
\multirow{4}{*}{}  $\Sigma=\Sigma_2$      &0   &0.048 &0.064       & 0.056     &0.064  & 0.064      &0.052      &0.048      &0.064    \\
                      & 0.25  &  1     & 0.288     & 0.200     & 0.168  & 1 &0.296     & 0.244     &0.196   \\
                      & 0.35  &   1    &  0.860    & 0.596     & 0.404  & 1      &0.876      & 0.532     & 0.380    \\
                      & 0.45  &    1   & 1     & 0.980     & 0.820 &1       & 1     &  0.956    & 0.796    \\ \hline
\end{tabular}
\caption{Type I error (i.e. $c=0$) and the powers of the centralized test and distributed one for Example 2, where `Cen' stands for  the centralized test and the columns with different values of $k$ are of the distributed test. \label{Table 2}}
\end{center}
\end{table}

\subsection{Real data}
We compare the distributed and centralized tests on the Statlog (Landsat Satellite) dataset from the UCI machine learning repository \citep{Dua2019},  a   multi-spectral data  having  36 covariates  and 6 class labels. The covariates   contain the pixel values of the four spectral bands of each of the 9 pixels in the $3\times3$ neighbourhood in a satellite image, and the response is  the classification label of the central pixel.  This  data have been analyzed by several authors \citep[etc.]{Tang2005,Kim2012}.  In multi-spectral data analysis,    predictors in each class are generally viewed as  normal vectors  in the literature of geography and biology  \cite[etc.]{Pickup1993,Bauer2011,Natsagdorj2017}.    For this data, classes 3 and 7 represent  grey soil and very damp grey soil with the sample size being 1357 and 1508, respectively. After applying  the Box-Cox transformations on variables of indices 9, 24, 35, and 36, we find that the data in each  class   are approximately normal.
 Denoting  by $\bm\mu_g$ the mean of class $g$ with $g\in\{3,7\}$, we aim to  test $H_0: \bm\mu_3-\bm\mu_7=0$.

We select $n=1000$ observations at random  for each class, denoted as $\{X_{gi},i=1,\cdots,n, g=3,7\}$, and     construct the observations  $Z_{i}=X_{3i}-X_{7i}, 1\le i\le n$.
The problem we  considered becomes  $H_0: E(Z_{i})=0$. We apply the centralized Hotelling $T^2$ test and the distributed  one  with $k=15$ and $25$, respectively, and compute the $p$-values. Repeating  the procedure 100 times,  the  results of the average of the  $p$-values   are presented in the column associated with  $\delta=0$  in Table \ref{Table 3}.

Furthermore, to compare the powers of the centrialized test with those of the distributed one, we shift the means such that two populations are close to each other.  Reminding   that there are 357 and 508 observations  left   for classes 3 and 7, respectively, then  for $g=3$ and 7, $\bm\mu_g$ can be estimated by the sample mean of these remaining  observations in class $g$, denoted as $\bar{\bm\mu}_g$, which is independent of $Z_i$'s.  Then we  shift the mean $\bm\mu_3$ of class 3 to $\bm\mu_3-\delta(\bar{\bm\mu}_3-\bar{\bm\mu}_7)$ with $0\le \delta<1$, which is equivalent to shifting  $Z_i$ to  $Z_{i,\delta}=Z_i-\delta(\bar{\bm\mu}_3-\bar{\bm\mu}_7)$. Clearly, when   $ \bm\mu_3\ne\bm\mu_7$,     distinguishing  two classes   becomes more difficult for large $\delta$. We repeat the procedure for 100 times and compute the average of the $p$-values. The results are presented in Table 3.  It is seen that the centralized test performs better when $\delta$ is close to 1.  

\begin{table}[htb]
\begin{center}
\begin{tabular}{llllllll}
\hline
Methods                                 &        & \multicolumn{6}{c}{$\delta$}   \\ \cline{3-8}
                                        &        &0     & 0.3  & 0.6    & 0.7       & 0.8   & 0.9 \\ \hline
 Centralized                            &        &0     &0   &0   &0       &0          &1.4e-8          \\
 Distributed                            &$k=15$  &0  &0    &0     &0       &2.2e-6          &1.6e-2    \\
                                        & $k=25$ &0    &0     &5.6e-5&3.2e-3   &0.0649 &0.1474     \\ \hline
\end{tabular}
\caption { The average of $p$-values over 100  replicas, where 1.4e-8 represents $1.4\times 10^{-8}$ and other quantities are defined similarly.\label{Table 3} }
\end{center}
\end{table}

\section{Discussion}
In this paper, we consider the one sample mean testing in distributed frameworks, extending Hotelling $T^2$ test and the statistic of \cite{Wang2015}, respectively. Obviously, the divide and conquer technique can be applied to many other statistics for one sample mean test, including  those, are developed for high dimensional cases. For distributed tests constructed from the divide and conquer technique, it is observed that there is a fundamental tradeoff between the communication costs and the powers of the test.

\section{Appendix}
\subsection{Proof of Theorem \ref{Th1}}
Recalling that
$T_l ^2\sim T^2 (p,n_l-1)$ and using the connection between Hotelling $T^2$ distribution and the $F$-distribution, we have
\[\frac{n_l -p}{(n_l -1)p}T_l ^2\sim
F(p,n_l -p),\qquad  l=1,\cdots, k.\]
Then, it is easy to derive that
\[E(T_l ^2/\sqrt p)=\frac{(n_l -1)\sqrt p}{n_l -p-2},\quad \var(T_l ^2 /\sqrt p)=\frac{2(n_l -1)^2(n_l -2)}{(n_l -p-2)^2(n_l -p-4)}.\]
We now establish the asymptotic normality of $\sqrt k T_{dis,n}^2$ by verifying the following two conditions (B1) and (B2) required by the Lyapunov  central limit theorem  \citep{Billingsley2008}.

\begin{itemize}
  \item[ (B1)]  $\var( T_l^2/\sqrt{p})$ is finite.
  \item[ (B2)]  Lyapunov condition: for some $\delta>0$,
\[\lim_{k\to\infty}\frac{1}{s_k^{2+\delta}}\sum_{l=1}^k E\left[\left|T_l^2 / \sqrt p-E\left(T_l^2 / \sqrt p\right)\right|^{2+\delta}\right]=0,\]
where $s_k=\big[\sum_{l=1}^k \mathrm{var}(T_l^2/\sqrt p)\big]^{1/2}.$
\end{itemize}
It is easy to see that (B1) holds. In fact, it follows from (A1) that
\[\var(\sqrt k T_{dis,n}^2)=\var\left(T_l ^2/\sqrt p\right)=\frac{2(n_l -1)^2(n_l -2)}{(n_l -p-2)^2(n_l -p-4)} \to \frac{2}{(1-r)^3}.\]
Thus $\sqrt k T_{dis,n}^2$ has finite variance as $r\in[0,1)$.

We now verify  condition (B2) for  $\delta=2$.
For a positive integer $\delta$,  denote   order-$\delta$ central moment of distribution $F(p,n_l-p)$ as $a_\delta$. Recalling  that $(n_l -p)[(n_l -1)p]^{-1}T_l ^2\sim F(p,n_l-p)$, then the order-4 central moment of $T_l ^2/\sqrt p$ can be written as
\[E\left[\left|T_l^2 / \sqrt p-E\left(T_l^2 / \sqrt p\right)\right|^4\right]=\left[\frac{(n_l-1)\sqrt p}{n_l-p}\right]^4 a_4.\]
Moreover, due to  the fact that the  kurtosis $K$ of  $F(p,n_l-p)$ equals $a_4/a^2_2$,  it follows that
\[a_4=Ka_2^2=K\left[\frac{2(n_l-p)^2(n_l-2)}{p(n_l-p-2)^2(n_l-p-4)}\right]^2.\]
  When $n_l-p>8$, it holds that
 \[K=\frac{12\{p[5(n_l-p)-22](n_l-2)+(n_l-p-4)(n_l-p-2)^2\}}{p(n_l-p-6)(n_l-p-8)(n_l-2)}+3.\]
Under condition (A1), we have $K=3+O(1/p)$ and consequently
\[E\left[\left|T_l^2 / \sqrt p-E\left(T_l^2 /\sqrt p \right)\right|^4\right]=O(1).\]
Thus,
\[\frac{1}{s_k^4}\sum_{l=1}^kE\left[\left|T_l^2 / \sqrt p-E\left(T_l^2 / \sqrt p\right)\right|^4\right]=O(1/k).\]
This verifies the condition (B2).
As $k\to \infty$, the conclusion is derived by the Lyapunov central limit theorem. This completes the proof.

\subsection{Proof of Theorem \ref{Th2}}
(1) We first prove the conclusions on the distributed test statistic.

Recall the asymptotic distribution (\ref{T_p-norm}) of $\sqrt k T_{dis,n}^2$ under $H_0$.  Let
\begin{equation}
E_0=\frac{(n_l-1)\sqrt p}{n_l-p-2},\qquad
V_0^2=\frac 2 {(1-r)^3}.\label{e1}
\end{equation}
Under $H_1$, the power function of distributed test has the following  form,
\[\psi_{dis}=P\left(V_0^{-1}\left|\sqrt {k} \left( T_{dis,n}^2-E_0\right)\right|>z_{1-\frac{\alpha}{2}}  \bigg |   H_1\right),\]
where $ T_{dis,n}^2 =  k^{-1}\sum_{l=1}^k T_l^2/\sqrt p$.
Under $H_1$, $(n_l -p)[(n_l -1)p]^{-1}T_l ^2$ follows a  non-central $F$-distribution $F(p,n_l-p,n_l\Delta)$ for $l=1,\cdots,k$. Thus it follows that
\begin{equation}
E(T_l^2)=\frac{(n_l-1)(p+n_l\Delta)}{n_l-p-2},\qquad
\var(T_l^2)=2(n_l-1)^2\frac{(p+n_l\Delta)^2+(p+2n_l\Delta)(n_l-p-2)}{(n_l-p-2)^2(n_l-p-4)}.\notag
\end{equation}
Under condition (A1),  it holds that, as $n\to\infty$,
\[\frac{\var(T_l^2)}  {p+n_l\Delta}\to \frac{2(\Delta^2+2\Delta+r)} {(r+\Delta)(1-r)^3}.\]
And we denote
$$V_1^2=\frac{2(\Delta^2+2\Delta+r)} {(r+\Delta)(1-r)^3}.$$
Recall $n_l=n/k, l=1,\cdots,k$.  Let $b_{n,\Delta}=\sqrt{p/(p+n_l\Delta)}$ and define
\[\widetilde{T}_{dis,n}^2:=b_{n,\Delta} T_{dis,n}^2= \frac{1}{k}\sum_{l=1}^k \frac{T_l^2}{\sqrt{p+n_l\Delta}},\]
which is a sum of  $i.i.d.$ random variables.  
The normality of $\widetilde{T}_{dis,n}^2$ can be derived by verifing the Lyapunov condition: for some $\delta>0$,
\begin{equation}\label{Condition}
\lim_{k\to\infty}\frac{1}{\tilde s_k^{2+\delta}}\sum_{l=1}^k E\left[\left|T_l^2 / \sqrt{p+n_l\Delta}-E\left(T_l^2 / \sqrt{p+n_l\Delta}\right)\right|^{2+\delta}\right]=0,
\end{equation}
where $\tilde s_k=\big[\sum_{l=1}^k \mathrm{var}(T_l^2/\sqrt{p+n_l\Delta})\big]^{1/2}.$ 
For $\delta=2$, we have $\tilde s_k^4=k^2 V_1^4(1+o(1))$, as $n\to\infty$.
Recall that $(n_l -p)[(n_l -1)p]^{-1}T_l ^2$ follows a  non-central $F$-distribution. Then following \cite{Patnaik1949}, the first four moments about the origin of the non-central $F$-distribution, we can derive that $E(|T_l^2 / \sqrt{p+n_l\Delta}-E(T_l^2 / \sqrt{p+n_l\Delta})|^{4})=O(1)$, as $n\to\infty$. Then we have
\[\tilde s_k^4\sum_{l=1}^k E\left[\left|T_l^2 / \sqrt{p+n_l\Delta}-E\left(T_l^2 / \sqrt{p+n_l\Delta}\right)\right|^4\right]=O(1/(kV_1^4)).\]
This verifies (\ref{Condition}).
Then as $k\to \infty$, we obtain the asymptotic normality
\[\sqrt k \left(\widetilde{T}_{dis,n}^2-E_1\right)   \stackrel{d}{\longrightarrow}N\left(0,V_1^2\right),\]
where
\begin{equation}
E_1=\frac{E(T_l^2)}{\sqrt{p+n_l\Delta}}=\frac{ (n_l-1)(p+n_l\Delta)}{\sqrt p(n_l-p-2)}.\label{e2}
\end{equation}
Under $H_1$, as $\min\{k,n\}\to \infty$, the power function of the distributed test has the following form,
\begin{align}\label{phi-dis}
\psi_n^{dis}&=P\left(V_0^{-1}\left|\sqrt k \left( T_{dis,n}^2-E_0\right)\right|>z_{1-\frac{\alpha}{2}}  \bigg |  H_1\right) \notag \\
&=P\left(V_1^{-1}\left|\sqrt k \left(\widetilde{T}_{dis,n}^2-E_1\right)+\sqrt k \left(E_1-b_{n,\Delta}E_0\right)\right|>V_1^{-1}V_0 b_{n,\Delta} z_{1-\frac{\alpha}{2}}  \bigg |  H_1\right) \notag \\
&\to 1-\Phi\left(V_1^{-1}V_0 b_{n,\Delta} z_{1-\frac{\alpha}{2}}-V_1^{-1} \sqrt {k}\left(E_1-b_{n,\Delta}E_0\right)\right)  \notag \\
&\qquad+\Phi\left(V_1^{-1}V_0 b_{n,\Delta} z_{\frac{\alpha}{2}}-V_1^{-1} \sqrt {k}\left(E_1-b_{n,\Delta}E_0\right)\right).
\end{align}
Substituting  (\ref{e1}) and (\ref{e2}) into (\ref{phi-dis}), the right hand side of (\ref{phi-dis})  can be denoted as
\[\phi_{\Delta,\gamma_{n}}=1-\Phi\left(A_{1n}z_{1-\frac{\alpha}{2}}-A_{2n}\sqrt{(1-\gamma_n)/2}\right)+\Phi\left(A_{1n}z_{\frac{\alpha}{2}}-A_{2n}\sqrt{(1-\gamma_n)/2}\right),\]
where
\[A_{1n}=\left(\frac{\gamma_{n}}{\gamma_{n}+2\Delta+\Delta^2}\right)^{1/2}, \qquad A_{2n}=\left(\frac {n\Delta^2}{\gamma_{n}+2\Delta+\Delta^2}\right)^{1/2}.\]
Then we have $\psi_n^{dis}-\phi_{\Delta,\gamma_{n}}\to0$.

(2) We prove the conclusions on the centralized test statistic.

First denoting
\[\widetilde T_{cen,n}^2:=\frac{n-p}{(n-1)p}T_{cen,n}^2,\]
the power function of the centralized test under $H_1$ can be written as
\begin{align}\label{eqcen}
\psi_n^{cen}&=P\left(\widetilde T_{cen,n}^2>F_{1-\alpha}(p,n-p)\Big| H_1\right).
\end{align}
Since $\widetilde T_{cen,n}^2$ follows the non-central $F$-distribution $F(p,n-p,n\Delta)$ under $H_1$,
  it is easy to see that
\[E\left(\widetilde T_{cen,n}^2\right)=\frac{n-p}{(n-p-2)p}(p+n\Delta),\]
\[\var\left(\widetilde T_{cen,n}^2\right)=\frac{2(n-p)^2}{(n-p-2)^2}\left[\frac{(p+n\Delta)^2}{(n-p-4)p^2}+\frac{(n-p-2)(p+2n\Delta)}{(n-p-4)p^3}\right].\]
Then denote
\[\widetilde V:=\left[\var\left(\widetilde T_{cen,n}^2\right)\right]^{1/2}\]
as the standard deviation of $F(p,n-p,n\Delta)$.

Thus the power of centralized test statistic $\psi_n^{cen}$ in (\ref{eqcen}) can be written as
\begin{align}
\psi_n^{cen}&=P\left(\widetilde V ^{-1}\widetilde T_{cen,n}^2>\widetilde V ^{-1}F_{1-\alpha}(p,n-p)\Big|H_1\right).  \notag
\end{align}
As $n\to\infty$, a variable following   $F(p,n-p)$ can be denoted as $1+ O_p(1/\sqrt p)$. Therefore $F_{1-\alpha}(p,n-p)$ can be written as $1+c/\sqrt p$, where $c$ is a bounded constant.  Then when $\Delta\gg\sqrt p/n$, one  can verify easily that, 
\begin{align}\label{ine}
E\left(\widetilde V ^{-1}\widetilde T_{cen,n}^2\right)\gg
\widetilde V ^{-1}F_{1-\alpha}(p,n-p),\notag
\end{align}
and that the variance of $\widetilde V ^{-1}\widetilde T_{cen,n}^2$ equals 1. Therefore,  the centralized test power $\psi_n^{cen}$ tends to 1, as $n\to \infty$.  $\blacksquare$

\subsection{Proof of Theorem \ref{thmh}}

According to  \cite{Wang2015}, it follows that
$$\var(G_{l})=n_l(n_l-1){\rm Tr}(B_\epsilon^2)/2:=V^2_{n_l}, \qquad  l=1,\cdots, k.$$
Denote $$E_{n_l}=n_l(n_l-1)(\bm\mu-\bm\mu_0)^\top A_\epsilon^2(\bm\mu-\bm\mu_0)(1+o(1))/2.$$


(1) We prove the conclusion under $H_0$.

 Under $H_0$, $G_{l}=\sum_{i,j\in S_l,j<i}Z^\top_i Z_j, 1\le l\le k$ are $i.i.d.$ variables, satisfying  $E(G_{l})=0$ and $\var(G_{l})=V_{n_l}^2$. Recall  that $p=p_n$, $k=k_n$, and $n_l=n/k$ are functions of $n$. Then  $\{G_{l},1\le l\le k, n_l=1,2,\cdots\}$ is a double array.
To  derive  the  asymptotic normality  of
$$\frac{G_{dis,n}}{\sqrt{k}V_{n_l}}=\frac{1}{\sqrt{k}}\sum_{l=1}^k \frac{G_{l}}{V_{n_l}},$$
we apply   the central limit theorem of double array  \citep{Anderson1984}. To this end, it is sufficient to check the following  condition:
\begin{equation}\label{Con}
\sum_{l=1}^k E\left(\left|\frac{G_{l}}{V_{n_l}}\right|^{2+\delta}\right)=o(S_k^{2+\delta}),\quad n\to\infty,
\end{equation}
for some $\delta>0$, where $S^2_k=\sum_{l=1}^k\var(G_{l}/V_{n_l})=k$.
We will verify   (\ref{Con})   with  $\delta=2$, that is,
\begin{equation}\label{Con2}
E\left(\left|\frac{G_{l}}{V_{n_l}}\right|^{4}\right)=o(k), \quad n\to\infty.
\end{equation}

It is sufficient to check (\ref{Con2}) for  $l=1$, since $G_{l}$'s are $i.i.d.$ variables.
For simplicity, we denote  the indices of the observations in the first machine as $\{1,2,\cdots, n_1\}$. Denote $G_{1}=\sum_{i=2}^{n_1}Y_i$, where $Y_i=\sum_{j=1}^{i-1} Z_i^\top Z_j$. Noting that $E(Z_i)=0$ and $E(Z_i^\top Z_j)=0$ for $i\ne j$ under  $H_0$, and consequently that  $E(Y_{i_1}Y_{i_2}Y_{i_3}Y_{i_4})=0$ for any indices  $i_1\ne i_2\ne i_3\ne i_4$, then it follows that
\begin{align}
E(|G_{l}|^4) &=\sum_{i_1=2}^{n_1} E(Y_{i_1}^4)+\sum_{\substack{i_1,i_2=2,\\ i_1\neq i_2}}^{n_1}E(Y_{i_1}Y_{i_2}^3)+ \sum_{\substack{i_1,i_2=2,\\ i_1\neq i_2}}^{n_1}E(Y_{i_1}^2 Y_{i_2}^2)+\sum^{n_1}_{\substack{i_1,i_2,i_3=2, \\ i_1\neq i_2\neq i_3}}E(Y_{i_1}Y_{i_2}Y_{i_3}^2). \notag
\end{align}
Under the conditions  ($i$) and ($ii$) in (A2), following \cite{Wang2015}, we have $\sum_{i_1=2}^{n_1} E(Y_{i_1}^4)/V_{n_1}^4 \to 0$, as $n_1$, $p\to\infty$.
Noting that $Z_iZ_j$ and $Z_iZ_s$ for any $i\ne j\ne s$ follow the same  distribution, by applying  the Cauchy-Schwartz inequality, it follows  that
\begin{align}
\sum_{\substack{i_1,i_2=2,\\ i_1\neq i_2}}^{n_1}E(Y_{i_1}Y_{i_2}^3)&=6\sum_{i_1<i_2}(i_1-1)E\{Z_1^\top Z_2 Z_1^\top Z_3 (Z_2^\top Z_3)^2\}  \notag\\
&\le 6\sum_{i_1<i_2}(i_1-1)E\{(Z_1^\top Z_2)^4\} \le 6n_1^3E\{(Z_1^\top Z_2)^4\}, \notag
\\
\sum_{\substack{i_1,i_2=2,\\ i_1\neq i_2}}^{n_l}E(Y_{i_1}^2 Y_{i_2}^2)=&8\sum_{i_1<i_2}(i_1-1)(i_2-2)E(Z_1^\top Z_2 Z_1^\top Z_3 Z_2^\top Z_4 Z_3^\top Z_4)  \notag  \\
&+4\sum_{i_1<i_2}(i_1-1)\Big[E\{(Z_1^\top Z_2)^2  (Z_3^\top Z_4)^2\}
+E\{(Z_1^\top Z_2)^2 Z_1^\top Z_3 Z_2^\top Z_3\}   \notag\\
&+E\{(Z_1^\top Z_2)^2  (Z_2^\top Z_3)^2\}\Big] \notag\\
\le &Cn_1^4E\{(Z_1^\top Z_2)^4\}, \qquad \text{for some constant~} C>0,  \notag
\\
\sum^{n_l}_{\substack{i_1,i_2,i_3=2, \\ i_1\neq i_2\neq i_3}}E(Y_{i_1}Y_{i_2}Y_{i_3}^2)&=8\sum_{i_1<i_2<i_3}(i_1-1)E(Z_1^\top Z_2 Z_3^\top Z_2 Z_4^\top Z_1 Z_4^\top Z_3)\notag\\
&\le 8n_1^4E\{(Z_1^\top Z_2)^4\}.  \notag
\end{align}
Under the conditions  ($i$) and ($ii$) in (A2), and by Lemma 1 in \cite{Wang2015}, we have $E\{(Z_1^\top Z_2)^4\}=O(E^2\{(Z_1^\top Z_2)^2\})=O({\rm{Tr}}^2(B^2))$. Then as $\min\{n_l,p,k\}\to\infty$,
$$E\left(\left|\frac{G_{1}}{V_{n_1}}\right|^{4}\right)=O(1)=o(k).$$
Therefore, (\ref{Con2}) holds.
Then by central limit theorem (CLT) of a double array, as $\min\{n_l,k,p\}\to \infty$, we have
$$\frac{G_{dis,n}}{\sqrt{n(n/k-1){\rm{Tr}(B_\epsilon^2)/2}}}=\frac{1}{\sqrt k}\sum_{l=1}^k \frac{G_{l}}{V_{n_l}}\stackrel{d}{\longrightarrow} N(0,1).$$

(2) Under $H_1$, following \cite{Wang2015}, it holds that
\begin{align}
G_{l}&=\sum_{\substack{i,j\in S_l,\\j<i}}\left\{\frac{\epsilon_i}{\|\epsilon_i\|}+\left(\frac{\epsilon_i+\bm\mu-\bm\mu_0}{\|\epsilon_i+\bm\mu-\bm\mu_0\|}-\frac{\epsilon_i}{\|\epsilon_i\|}\right)\right\}^\top
\left\{\frac{\epsilon_j}{\|\epsilon_j\|}+\left(\frac{\epsilon_j+\bm\mu-\bm\mu_0}{\|\epsilon_j+\bm\mu-\bm\mu_0\|}-\frac{\epsilon
_j}{\|\epsilon_j\|}\right)\right\}      \notag\\
&=G_{l,1}+G_{l,2}+G_{l,3},\notag
\end{align}
where
\begin{align}
G_{l,1}&=\sum_{i,j\in S_l, j<i}\frac{\epsilon_i^\top\epsilon_j}{\|\epsilon_i\| \|\epsilon_j\|},\notag   \qquad
G_{l,2}=\sum_{i,j\in S_l, j\neq i}\left(\frac{\epsilon_i+\bm\mu-\bm\mu_0}{\|\epsilon_i+\bm\mu-\bm\mu_0\|}-\frac{\epsilon_i}{\|\epsilon_i\|}\right)^\top \frac{\epsilon_j}{\|\epsilon_j\|},  \notag \\
G_{l,3}&=\sum_{i,j\in S_l, j<i}\left(\frac{\epsilon_i+\bm\mu-\bm\mu_0}{\|\epsilon_i+\bm\mu-\bm\mu_0\|}-\frac{\epsilon_i}{\|\epsilon_i\|}\right)^\top\left(\frac{\epsilon_j+\bm\mu-\bm\mu_0}{\|\epsilon_j+\bm\mu-\bm\mu_0\|}-\frac{\epsilon_j}{\|\epsilon_j\|}\right). \notag
\end{align}
By the conclusion  (1)  in this subsection, we have
$$\frac{1}{\sqrt k}\sum_{l=1}^k G_{l,1}/V_{n_l}\stackrel{d}{\rightarrow}  N(0,1).$$
Moreover, by conditions (A1) and (A4$'$) in Section 3.2, and Lemma 1, 4, and 5 of  \cite{Wang2015}, we have
 $E(G_{l,2})=0$ and $\var(G_{l,2}/V_{n_l})=o(1)$. As a result, it follows that  $\var\left(\sum_{l=1}^k G_{l,2}/(\sqrt{k}V_{n_l})\right)=o(1)$, that is, $\sum_{l=1}^k G_{l,2}/(\sqrt{k}V_{n_l})=o_p(1)$.
For $G_{l,3}$, denote
\begin{align}
G_{l,31}=&\frac{n_l(n_l-1)}{2}E\left(\frac{\epsilon_1+\bm\mu-\bm\mu_0}{\|\epsilon_1+\bm\mu-\bm\mu_0\|}\right)^\top E\left(\frac{\epsilon_2+\bm\mu-\bm\mu_0}{\|\epsilon_2+\bm\mu-\bm\mu_0\|}\right),  \notag\\
G_{l,32}=&\sum_{i,j\in S_l,j\neq i}E\left( \frac{\epsilon_i+\bm\mu-\bm\mu_0}{\|\epsilon_i+\bm\mu-\bm\mu_0\|}\right)^\top\left\{\frac{\epsilon_j+\bm\mu-\bm\mu_0}{\|\epsilon_j+\bm\mu-\bm\mu_0\|}-\frac{\epsilon_j}{\|\epsilon_j\|}-E\left( \frac{\epsilon_j+\bm\mu-\bm\mu_0}{\|\epsilon_j+\bm\mu-\bm\mu_0\|}\right)\right\},\notag\\
G_{l,33}=&\sum_{i,j\in S_l,j<i}\left\{\frac{\epsilon_i+\bm\mu-\bm\mu_0}{\|\epsilon_i+\bm\mu-\bm\mu_0\|}-\frac{\epsilon_i}{\|\epsilon_i\|}-E\left( \frac{\epsilon_i+\bm\mu-\bm\mu_0}{\|\epsilon_i+\bm\mu-\bm\mu_0\|}\right)\right\}^\top  \notag\\
&\times
\left\{\frac{\epsilon_j+\bm\mu-\bm\mu_0}{\|\epsilon_j+\bm\mu-\bm\mu_0\|}-\frac{\epsilon_j}{\|\epsilon_j\|}-E\left( \frac{\epsilon_j+\bm\mu-\bm\mu_0}{\|\epsilon_j+\bm\mu-\bm\mu_0\|}\right)\right\}.  \notag
\end{align}
Then it follows that
\begin{align}
G_{l,3}=&\sum_{i,j\in S_l,j\neq i}\left[E\left(\frac{\epsilon_i+\bm\mu-\bm\mu_0}{\|\epsilon_i+\bm\mu-\bm\mu_0\|}\right)+\left\{\frac{\epsilon_i+\bm\mu-\bm\mu_0}{\|\epsilon_i+\bm\mu-\bm\mu_0\|}-\frac{\epsilon_i}{\|\epsilon_i\|}-E\left( \frac{\epsilon_i+\bm\mu-\bm\mu_0}{\|\epsilon_i+\bm\mu-\bm\mu_0\|}\right)\right\}\right]^\top  \notag\\
&\times
\left[E\left(\frac{\epsilon_j+\bm\mu-\bm\mu_0}{\|\epsilon_j+\bm\mu-\bm\mu_0\|}\right)+\left\{\frac{\epsilon_j+\bm\mu-\bm\mu_0}{\|\epsilon_j+\bm\mu-\bm\mu_0\|}-\frac{\epsilon_j}{\|\epsilon_j\|}-E\left( \frac{\epsilon_j+\bm\mu-\bm\mu_0}{\|\epsilon_j+\bm\mu-\bm\mu_0\|}\right)\right\}\right]   \notag\\
=&G_{l,31}+G_{l,32}+G_{l,33}.\notag
\end{align}
By conditions (A2), (A3), and (A4$'$) in Section 3.2, Lemma 4 and 6 of \cite{Wang2015}, we have $G_{l,31}=E_{n_l}$, which  is a constant. Also, it is easy to see $E(G_{l,32})=E(G_{l,33})=0$,
$\var(G_{l,32}/V_{n_l})=o(1)$, and $\var(G_{l,33}/V_{n_l})=o(1).$
Therefore, $\sum_{l=1}^k (G_{l,3}-E_{n_l})/(\sqrt{k}V_{n_l})=o_p(1)$.
Then we have
$$\frac{G_{dis,n}-kE_{n_l}}{\sqrt k V_{n_l}}=\frac{1}{\sqrt k}\sum_{l=1}^k (G_{l}-E_{n_l})/V_{n_l}=\frac{1}{\sqrt k}\sum_{l=1}^k G_{l,1}/V_{n_l}+o_p(1).$$
Then by CLT and Slutsky's theorem, the asymptotic distribution of $G_{dis,n}$ can be derived.
$\blacksquare$

\section*{Funding}
Junlong Zhao was supported by National Science Foundation of China, No. 11871104, and the
Fundamental Research Funds for the Central Universities, China.

\bibliographystyle{elsarticle-harv}

\bibliography{bibfile}

\begin{thebibliography}{42}
\expandafter\ifx\csname natexlab\endcsname\relax\def\natexlab#1{#1}\fi
\expandafter\ifx\csname url\endcsname\relax
  \def\url#1{\texttt{#1}}\fi
\expandafter\ifx\csname urlprefix\endcsname\relax\def\urlprefix{URL }\fi

\bibitem[{Anderson(1984)}]{Anderson1984}
Anderson, T.~W., 1984. An introduction to multivariate statistical analysis,
  3rd Edition. World Scientific Publishing Co.

\bibitem[{Bai and Saranadasa(1996)}]{Bai1996}
Bai, Z., Saranadasa, H., 1996. Effect of high dimension: by an example of a two
  sample problem. Statistica Sinica 6~(2), 311--329.

\bibitem[{Battey et~al.(2018)Battey, Fan, Liu, Lu, and Zhu}]{Battey2018}
Battey, H., Fan, J., Liu, H., Lu, J., Zhu, Z., 2018. Distributed testing and
  estimation under sparse high dimensional models. The Annals of Statistics
  46~(3), 1352--1382.

\bibitem[{Bauer et~al.(2011)Bauer, Kor{\v{c}}, and F{\"o}rstner}]{Bauer2011}
Bauer, S.~D., Kor{\v{c}}, F., F{\"o}rstner, W., 2011. The potential of
  automatic methods of classification to identify leaf diseases from
  multispectral images. Precision Agriculture 12~(3), 361--377.

\bibitem[{Billingsley(2008)}]{Billingsley2008}
Billingsley, P., 2008. Probability and measure. John Wiley \& Sons.

\bibitem[{Chakraborty et~al.(2017)Chakraborty, Chaudhuri,
  et~al.}]{Chakraborty2017}
Chakraborty, A., Chaudhuri, P., et~al., 2017. Tests for high-dimensional data
  based on means, spatial signs and spatial ranks. The Annals of Statistics
  45~(2), 771--799.

\bibitem[{Chen et~al.(2011)Chen, Paul, Prentice, and Wang}]{Chen2011}
Chen, L.~S., Paul, D., Prentice, R.~L., Wang, P., 2011. {A regularized
  Hotelling’s $T^2$ test for pathway analysis in proteomic studies}. Journal
  of the American Statistical Association 106~(496), 1345--1360.

\bibitem[{Chen and Peng(2018)}]{Chen2018}
Chen, S.~X., Peng, L., 2018. Distributed statistical inference for massive
  data. arXiv preprint arXiv:1805.11214.

\bibitem[{Chen et~al.(2010)Chen, Qin, et~al.}]{Chen2010}
Chen, S.~X., Qin, Y.-L., et~al., 2010. A two-sample test for high-dimensional
  data with applications to gene-set testing. The Annals of Statistics 38~(2),
  808--835.

\bibitem[{Chen and Xie(2014)}]{Chen2014}
Chen, X., Xie, M.-g., 2014. A split-and-conquer approach for analysis of
  extraordinarily large data. Statistica Sinica, 1655--1684.

\bibitem[{Davison(1992)}]{Davison1992}
Davison, A., 1992. Treatment effect heterogeneity in paired data. Biometrika
  79~(3), 463--474.

\bibitem[{Dempster(1958)}]{Dempster1958}
Dempster, A.~P., 1958. A high dimensional two sample significance test. The
  Annals of Mathematical Statistics, 995--1010.

\bibitem[{Dobriban and Sheng(2021)}]{Dobriban2021}
Dobriban, E., Sheng, Y., 2021. Distributed linear regression by averaging. The
  Annals of Statistics 49~(2), 918--943.

\bibitem[{Dong et~al.(2016)Dong, Pang, Tong, and Genton}]{Dong2016}
Dong, K., Pang, H., Tong, T., Genton, M.~G., 2016. {Shrinkage-based diagonal
  Hotelling's tests for high-dimensional small sample size data}. Journal of
  Multivariate Analysis 143, 127--142.

\bibitem[{Dua and Graff(2017)}]{Dua2019}
Dua, D., Graff, C., 2017. {UCI} machine learning repository.
\newline\urlprefix\url{http://archive.ics.uci.edu/ml}

\bibitem[{Edward~Jackson(1985)}]{Edward1985}
Edward~Jackson, J., 1985. Multivariate quality control. Communications in
  Statistics-Theory and Methods 14~(11), 2657--2688.

\bibitem[{Fan et~al.(2019)Fan, Wang, Wang, and Zhu}]{Fan2019}
Fan, J., Wang, D., Wang, K., Zhu, Z., 2019. Distributed estimation of principal
  eigenspaces. The Annals of Statistics 47~(6), 3009--3031.

\bibitem[{Guo et~al.(2019)Guo, Lin, and Shi}]{Guo2019}
Guo, Z.-C., Lin, S.-B., Shi, L., 2019. Distributed learning with multi-penalty
  regularization. Applied and Computational Harmonic Analysis 46~(3), 478--499.

\bibitem[{Haug et~al.(2011)Haug, Huber, Schlabach, Becher, and
  Thomsen}]{Haug2011}
Haug, L.~S., Huber, S., Schlabach, M., Becher, G., Thomsen, C., 2011.
  Investigation on per-and polyfluorinated compounds in paired samples of house
  dust and indoor air from {Norwegian} homes. Environmental Science \&
  Technology 45~(19), 7991--7998.

\bibitem[{Hsieh et~al.(2014)Hsieh, Si, and Dhillon}]{Hsieh2014}
Hsieh, C.-J., Si, S., Dhillon, I.~S., 2014. A divide-and-conquer solver for
  kernel support vector machines. In: International Conference on Machine
  Learning. pp. 566--574.

\bibitem[{Jordan et~al.(2019)Jordan, Lee, and Yang}]{Michael2019}
Jordan, M.~I., Lee, J.~D., Yang, Y., 2019. Communication-efficient distributed
  statistical inference. Journal of the American Statistical Association
  114~(526), 668--681.

\bibitem[{Kim and Ghahramani(2012)}]{Kim2012}
Kim, H.-C., Ghahramani, Z., 2012. Bayesian classifier combination. In:
  Artificial Intelligence and Statistics. PMLR, pp. 619--627.

\bibitem[{Lee et~al.(2012)Lee, Lim, Li, Vannucci, and Petkova}]{Lee2012}
Lee, S.~H., Lim, J., Li, E., Vannucci, M., Petkova, E., 2012. Order test for
  high-dimensional two-sample means. Journal of Statistical Planning and
  Inference 142~(9), 2719--2725.

\bibitem[{Li et~al.(2020)Li, Aue, Paul, Peng, Wang, et~al.}]{Li2020}
Li, H., Aue, A., Paul, D., Peng, J., Wang, P., et~al., 2020. {An adaptable
  generalization of Hotelling's $T^2$ test in high dimension}. The Annals of
  Statistics 48~(3), 1815--1847.

\bibitem[{Li and Zhao(2020)}]{Limengyu2020}
Li, M., Zhao, J., 2020. Communication-efficient distributed linear discriminant
  analysis for binary classification. Statistica Sinica.

\bibitem[{Lin et~al.(2017)Lin, Guo, and Zhou}]{Lin2017}
Lin, S.-B., Guo, X., Zhou, D.-X., 2017. Distributed learning with regularized
  least squares. Journal of Machine Learning Research 18~(1), 3202--3232.

\bibitem[{Natsagdorj et~al.(2017)Natsagdorj, Renchin, Kappas, Tseveen, Dari,
  Tsend, and Duger}]{Natsagdorj2017}
Natsagdorj, E., Renchin, T., Kappas, M., Tseveen, B., Dari, C., Tsend, O.,
  Duger, U.-O., 2017. An integrated methodology for soil moisture analysis
  using multispectral data in mongolia. Geo-spatial Information Science 20~(1),
  46--55.

\bibitem[{Patnaik(1949)}]{Patnaik1949}
Patnaik, P.~B., 1949. {The non-central $\chi^2$- and $F$-distributions and
  their applications}. Biometrika~(1-2), 1--2.

\bibitem[{Pickup et~al.(1993)Pickup, Chewings, and Nelson}]{Pickup1993}
Pickup, G., Chewings, V., Nelson, D., 1993. {Estimating changes in vegetation
  cover over time in arid rangelands using Landsat MSS data}. Remote sensing of
  Environment 43~(3), 243--263.

\bibitem[{Rubin(2006)}]{Rubin2006}
Rubin, D.~B., 2006. Matched sampling for causal effects. Cambridge University
  Press.

\bibitem[{Srivastava(2009)}]{Srivastava2009}
Srivastava, M.~S., 2009. A test for the mean vector with fewer observations
  than the dimension under non-normality. Journal of Multivariate Analysis
  100~(3), 518--532.

\bibitem[{Srivastava and Du(2008)}]{Srivastava2008}
Srivastava, M.~S., Du, M., 2008. A test for the mean vector with fewer
  observations than the dimension. Journal of Multivariate Analysis 99~(3),
  386--402.

\bibitem[{Srivastava et~al.(2016)Srivastava, Li, and Ruppert}]{Srivastava2016}
Srivastava, R., Li, P., Ruppert, D., 2016. {RAPTT: An exact two-sample test in
  high dimensions using random projections}. Journal of Computational and
  Graphical Statistics 25~(3), 954--970.

\bibitem[{Szab{\'o} and Van~Zanten(2019)}]{Szabo2017}
Szab{\'o}, B., Van~Zanten, H., 2019. An asymptotic analysis of distributed
  nonparametric methods. Journal of Machine Learning Research 20~(87), 1--30.

\bibitem[{Tang et~al.(2005)Tang, Suganthan, Yao, and Qin}]{Tang2005}
Tang, E.~K., Suganthan, P.~N., Yao, X., Qin, A.~K., 2005. {Linear
  dimensionality reduction using relevance weighted LDA}. Pattern Recognition
  38~(4), 485--493.

\bibitem[{Tian and Gu(2017)}]{Tian2017}
Tian, L., Gu, Q., 2017. Communication-efficient distributed sparse linear
  discriminant analysis. In: Artificial Intelligence and Statistics. PMLR, pp.
  1178--1187.

\bibitem[{Volgushev et~al.(2019)Volgushev, Chao, Cheng, et~al.}]{Volgushev2019}
Volgushev, S., Chao, S.-K., Cheng, G., et~al., 2019. Distributed inference for
  quantile regression processes. The Annals of Statistics 47~(3), 1634--1662.

\bibitem[{Wang et~al.(2015)Wang, Peng, and Li}]{Wang2015}
Wang, L., Peng, B., Li, R., 2015. A high-dimensional nonparametric multivariate
  test for mean vector. Journal of the American Statistical Association
  110~(512), 1658--1669.

\bibitem[{Xu et~al.(2016)Xu, Lin, Wei, and Pan}]{Xu2016}
Xu, G., Lin, L., Wei, P., Pan, W., 2016. An adaptive two-sample test for
  high-dimensional means. Biometrika 103~(3), 609--624.

\bibitem[{Ye et~al.(2002)Ye, Emran, Chen, and Vilbert}]{Ye2002}
Ye, N., Emran, S.~M., Chen, Q., Vilbert, S., 2002. Multivariate statistical
  analysis of audit trails for host-based intrusion detection. IEEE
  Transactions on Computers 51~(7), 810--820.

\bibitem[{Zhang et~al.(2015)Zhang, Duchi, and Wainwright}]{Zhang2015}
Zhang, Y., Duchi, J., Wainwright, M., 2015. Divide and conquer kernel ridge
  regression: A distributed algorithm with minimax optimal rates. Journal of
  Machine Learning Research 16~(1), 3299--3340.

\bibitem[{Zhang et~al.(2013)Zhang, Duchi, and Wainwright}]{Zhang2013}
Zhang, Y., Duchi, J.~C., Wainwright, M.~J., 2013. Communication-efficient
  algorithms for statistical optimization. The Journal of Machine Learning
  Research 14~(1), 3321--3363.

\end{thebibliography}

\end{document}